\documentclass[reqno]{amsart}
\usepackage{graphicx}

\newtheorem{lemma}{Lemma}
\newtheorem{theorem}{Theorem}

\newtheorem{hypothesis}{Hypothesis}
\newtheorem{remark}{Remark}

\newcommand{\bee}{\begin{eqnarray}}
\newcommand{\eee}{\end{eqnarray}}
\newcommand{\be}{\begin{eqnarray*}}
\newcommand{\ee}{\end{eqnarray*}}
\newcommand{\R}{{\mathbb R}}

\newcommand{\C}{{\mathbb C}}

\begin{document}
 
 \title [Perturbation Theory for NLS]{Perturbation theory for nonlinear Schr\"odinger equations}
 
 \author {
 Andrea Sacchetti
 }

\address {
Department of Physics, Informatics and Mathematics, University of Modena and Reggio Emilia, Modena, Italy.
}

\email {andrea.sacchetti@unimore.it}

\date {\today}

\thanks {{\it Acknowledgments.} \ This work is partially supported by the GNFM-INdAM and by the UniMoRe-FIM project ``Modelli e Metodi della Fisica Matematica''. \ We deeply thank Riccardo Adami for useful discussions and the referees for their insightful comments.}

\begin {abstract} Treating the nonlinear term of the Gross-Pitaevskii nonlinear Schr\"o\-din\-ger equation as a perturbation of an isolated discrete eigenvalue of the linear problem one obtains a Rayleigh-Schr\"odinger power series. \ This power series is proved to be convergent  when the parameter representing the intensity of the nonlinear term is less in absolute value than a threshold value, and it gives a stationary solution to the nonlinear Schr\"odinger equation. 

\bigskip

{\it Data availability statement.} \ All data generated or analysed during this study are included in this published article.

\bigskip

{\it Conflict of interest statement.} \ The author has no competing interests to declare that are relevant to the content of this article.


\end{abstract}


\maketitle

\section {Introduction} \label {Sec1}

Nonlinear Schr\"odinger equation (hereafter NLS) is a research topic with a large variety of applications \cite {Su}: from problems in nonlinear optics to the analysis of quantum dynamics of Bose-Einstein condensates. \ In particular, the study of its stationary solutions has attracted increasing attention, and, apart from the few cases in which the solution exists in explicit form, the analysis has mainly focused variational methods or 
on approximation methods based on both semiclassical and perturbative techniques. 

Variational methods are widely used in order to construct bound states for NLS with a linear potential, typically by solving a minimization problem; for instance, this is done by \cite {RW} where they proved the existence of a small amplitude stationary solution that bifurcates from the zero solution (see also \cite {SW1,SW2} where nonlinear scattering is considered for NLS with, respectively, one and two nonlinear bound states and the references there in). \ Similarly, in \cite {Aschbacher} variational methods have been applied to prove that for the minimizer of the nonlinear Hartree energy functional a symmetry breaking effect occurs.

For what concerns semiclassical methods, they have been successfully used in this framework where several authors have been able to demonstrate the existence, in the semiclassical limit  using variational techniques, of stationary solution concentrated around the critical points of the potential \cite {A,Fl,O,W}. \ Also the occurrence of bifurcation phenomena has been discussed in the semiclassical limit \cite {Sa}.

On the other side, the perturbative approach takes up the underlying idea of Rayleigh-Schr\"odinger series expansion, where the solution is written as a formal series of powers whose coefficients are determined recursively and where the convergence of the series is under investigation. \ Typically in these cases the perturbation is represented by the nonlinear term, and the unperturbed Schr\"odinger equation, where the nonlinear term is absent, admits isolated eigenvalues. \ Several applications of this idea have been developed over the years \cite {An1,An2,C,F,S,V} limited, in general, to a formal analysis of the series without proving its convergence. \ In fact, we should emphasize that the problem of convergence of the power series has been solved for some kind of nonlinear Schr\"odinger equations; more precisely, the spinless real Hartree-Fock model and the Thomas-Fermi-Von-Weizs\"acker model has been considered by \cite {C} proving, in particular, that in the first model the Rayleigh-Schr\"odinger perturbation series has a positive convergence radius. 

Finally, it should be mentioned that numerical methods based on discrete Galerkin approximations or spectral splitting methods are widely and effectively used for the study of time-dependent NLS (see \cite {Antoine,Bao1,Bao2,Sacchetti,Soffer} and references therein).

In this paper we aim to give a rigorous basis to the perturbative approach for computing the stationary solution of the NLS by going so far as to demonstrate, under fairly general assumptions, the convergence of the Rayleigh-Schr\"odinger series when the perturbative parameter, which measures the intensity of the nonlinear perturbation, is less in absolute value than a given threshold. \ In this way it is shown that the steady states associated with isolated and nondegenerate eigenvalues of the linear operator transform into stationary solutions of the NLS when nonlinearity is switched on, and the latter can be computed very efficiently through the convergent perturbative series. \ Finally, it is also possible to give a lower estimate of the radius of convergence of the power series.

The paper is organized as follows. \ In Section \ref {Sec2} we describe the model, we write the formal power series of the stationary solutions and we state the convergence result in Theorem \ref {Teo1}.  \ In Section \ref {Sec5} we state and prove some technical preliminary results. \ In Section \ref {Sec6} we obtain the convergence of the perturbative series proving thus Theorem \ref {Teo1}. \ In Sections \ref {Sec3} and \ref {Sec4} we discuss a couple of one-dimensional examples: namely in Section \ref {Sec3} we consider the case of an infinite well potential, in this case we are also able to compare the perturbative results with the exact ones; in Section \ref {Sec4} we compute the perturbative series in the case where the potential is the harmonic one. \ The discussion of these two models is, in some sense, ``pedagogical''; indeed, by means of numerical experiments it is possible to see that the coefficients of the power series expansion rapidly decreases and then one can guess the convergence radius of the power series. \ Finally, in Section \ref {Sec7} we draw some closing comments. \ A small technical appendix closes the paper.

\section {Main results} \label {Sec2}

\subsection {Assumptions}

We consider the time-independent nonlinear Schr\"odinger equation
\bee
H \psi + \nu |\psi |^2 \psi = E \psi \, , \ \psi \in L^2 (\R^d)  \, , \label {Eq1}
\eee
where $H = - \Delta +V$ is a linear operator formally defined on $L^2 (\R^d)$. \ The nonlinear term plays the role of perturbation and its strength $\nu \in \C$ is a small perturbative parameter.

\begin {hypothesis} \label {Hyp1}
The potential $V$ is assumed to be a real-valued piecewise continuous function bounded from below: 
\bee
V(x) \ge \Gamma \, , \ \forall x \in \R \, , \label {Eq2}
\eee
for some $\Gamma \in \R$.
\end {hypothesis}

\begin {remark}
We assume that the potential $V(x)$ is a piecewise continuous function bounded from below for the sake of simplicity. \ In fact, we must remark  that one could extend our treatment to the case where some milder conditions on $V(x)$ are assumed; however, we don't dwell on those details here. \ On the other side, it might be interesting to consider the case in which $V(x)$ is given by means of an attractive Dirac's $\delta$ \cite {DellaCasa}; this case does not fall under the Hypothesis \ref {Hyp1}.
\end {remark}

Hence, $H$ admits a self-adjoint extension, still denoted by $H$, on a self-adjointness domain ${\mathcal D}(H) \subset L^2 (\R^d )$.

\begin {hypothesis} \label {Hyp2}
The discrete spectrum of $H$ is not empty: $\sigma_d (H) \not= \emptyset$, and admits a non degenerate eigenvalue $e_0 \in \sigma_d (H)$  with associated eigenvector $\phi_0\in {\mathcal D} (H)$:
\bee
H\phi_0 = e_0 \phi_0 , \ \phi_0 \in L^2 (\R^d)   \, . \label {Eq3}
\eee
\end {hypothesis}

Hereafter we can assume, for simplicity's sake and without loss in generality, that the unperturbed eigenvector $\phi_0$ is normalized to one, i.e.:
\be
\| \phi_0 \|_{L^2} = 1\,. 
\ee

In the following we denote 
\be
\Lambda =  \mbox {\rm dist} \left [ \sigma (H) \setminus \{ e_0 \} , e_0 \right ] > 0 \, . 
\ee

\begin {remark} \label {Nota0}
Since $\phi_0 \in {\mathcal D}(H)$ and the potential $V$ is bounded from below then it follows that $\phi_0 \in H^1$ because 
\be
\| \nabla \phi_0 \|_{L^2}^2 = \langle - \Delta \phi_0 , \phi_0 \rangle_{L^2} = e_0 \| \phi_0 \|_{L^2}^2 - \langle V \phi_0 , \phi_0 \rangle\le (e_0 - \Gamma) \| \phi_0 \|_{L^2}^2\, . 
\ee
Thus
\be
\phi_0 \in {\mathcal D}(H) \cap L^6 (\R^d) 
\ee
follows from this fact and from the Gagliardo-Nirenberg inequality  \cite {Cordero}
\bee
\| f \|_{L^p} \le C_{p,d} \| \nabla f \|^\rho \| f \|^{1-\rho } \, , \ \rho = \frac d2 - \frac dp \label {Eq25}
\eee
for some positive constant $C_{p,d}$ and where
\bee
p \in 
\left \{ 
\begin {array}{l} 
\, [2,+\infty ] \mbox { if }d=1;\\
\, [2,+\infty ) \mbox { if }d=2;\\
\, [2, 2d/(d-2) ] \mbox { if }d > 2. 
\end {array}
\right. \, . \label {Eq26}
\eee
\end {remark}

\subsection {Formal solutions}

We look for a \emph {formal} stationary solution to (\ref {Eq1}) close to the solution to the linear problem (\ref {Eq3}) by means of a formal power series
\bee
E :=E(\nu ) = \lim_{N\to + \infty} E_N (\nu )\ \mbox { and } \ \psi := \psi (x, \nu ) = \lim_{N\to + \infty} \psi_N (x,\nu )  \, , \label {Eq4}
\eee
where
\bee
E_N (\nu ) = \sum_{n=0}^N \nu^n e_n \ \mbox { and } \ \psi_N (x, \nu )= \sum_{n=0}^N \nu^n \phi_n (x)  \label {Eq5}
\eee
and where $e_n$ and $\phi_n$ are defined by induction as follows. \ In fact, $E$ and $\psi$ depend on the perturbative parameter $\nu$; sometimes, for simplicity, we will omit this dependence when this fact does not cause misunderstanding.

\begin {remark} \label {Rem1}
We should underline that the following formulas make sense provided that the vectors $u_n$ and $v_n$ below belongs to $L^2 (\R^d)$ and $\phi_n \in {\mathcal D}(H) \cap L^6 (\R^d)$; we'll discuss this point in Section \ref {Sec5}.
\end {remark}

Let $e_\ell$ and $\phi_\ell$ be defined for any $\ell =0,1,\ldots , n-1$, where $\langle \phi_0 , \phi_\ell \rangle_{L^2} =0$ for any $\ell =1,2,\dots , n-1$, and let 
\be
v_{n-1} = \sum_{m=0}^{n-1} \sum_{\ell=0}^{n-1-m} \phi_m \bar \phi_\ell \phi_{n-1-m-\ell} \, , \ u_n = \sum_{m=1}^{n-1} e_m \phi_{n-m}\, . 
\ee
We define 
\bee
e_n =  {\langle \phi_0 , v_{n-1} \rangle_{L^2}}
\label {Eq6}
\eee
and 
\be
\varphi_n = e_n \phi_0 + u_n - v_{n-1}\, . 
\ee 
By construction it follows that 
\be
\langle u_n , \phi_0 \rangle_{L^2} =0 \ \mbox { and } \ \langle \varphi_n , \phi_0 \rangle_{L^2} =0\, ,
\ee
that is $\varphi_n \in \Pi^\perp L^2$, where $\Pi^\perp = 1 - \Pi$ and $\Pi$ is the projection operator on the space spanned by $\phi_0$. \ Hence, the resolvent operator $[H-e_0 ]^{-1}$ is bounded on $\Pi^\perp L^2$ and we can define
\bee
\phi_n = [H-e_0 ]^{-1} \varphi_n = [H-e_0 ]^{-1} \Pi^\perp \varphi_n = \Pi^\perp [H-e_0 ]^{-1}  \varphi_n \in \Pi^\perp L^2 \, . \label {Eq7}
\eee

\begin {lemma}\label {Lemma1}
Let $e_n$ and $\phi_n\in \Pi^\perp L^2 $ be defined by induction for any $n \ge 1$ as in (\ref {Eq6}) and (\ref {Eq7}). \ Let $E_N$ and $\psi_N$ be defined as in (\ref {Eq5}). \ Let
\bee
r_N := H \psi_N + \nu |\psi_N |^2 \psi_N - E_N \psi_N \, , \label {Eq8}
\eee
Then $r_N$ is a power series in $\nu$ with finitely many terms where all the coefficients of the powers $\nu^n$, with $n \le N$, are exactly zero.
\end {lemma}

\begin {remark} \label {Rem2}
Since $e_0$ is a simple and isolated eigenvalue of the selfadjoint operator $H$ and since $\varphi_n \perp \Pi L^2$ for any $n \ge 1$ then:
\bee
\| \phi_n \|_{L^2} \le \frac {1}{\Lambda } \| \varphi_n \|_{L^2} \, . 
\label {Eq9}
\eee
\end {remark}

\begin {proof} 
By formally substituting (\ref {Eq5}) and ( \ref {Eq4}) in (\ref {Eq1}) we then have to check that
\bee
\sum_{n=0}^\infty \nu^n H\phi_n + \nu \sum_{n=0}^\infty \nu^n \phi_n \sum_{m=0}^\infty \nu^m \phi_m \sum_{\ell =0}^\infty \nu^\ell \bar \phi_\ell 
= \sum_{m=0}^\infty \nu^m e_m  \sum_{n=0}^\infty \nu^n \phi_n \label {Eq10}
\eee

This equation can be written as 
\be
\sum_{n=0}^\infty \nu^n H \phi_n + \sum_{n=0}^\infty \nu^{n+1} v_n = \sum_{n=0}^\infty \nu^n \left [ e_0 \phi_n + u_n + e_n \phi_0 \right ] 
\ee
where $u_n$ and $v_n$ are defined above. \ By equating the term with the same power of the perturbative parameter $\nu$ we have that
\bee
H\phi_0 = e_0 \phi_0 \, , \ \mbox { for } \ n =0 \, , \label {Eq11}
\eee
which is satisfied by assumption, and
\be
H \phi_n + v_{n-1} = e_n \phi_0 + e_0 \phi_n + u_n \, , \ \mbox { for } \  n \ge 1 \, . 
\ee
If we multiply both side by $\phi_0$ then
\be
\langle \phi_0 , H \phi_n \rangle_{L^2} + \langle \phi_0 , v_{n-1} \rangle_{L^2} = e_n \| \phi_0 \|^2_{L^2} + e_0 \langle \phi_0 , \phi_n \rangle_{L^2} + \langle \phi_0 , u_n \rangle_{L^2} 
\ee
from which it follows that
\be
e_n = \frac {\langle \phi_0 , v_{n-1} \rangle_{L^2} - \langle \phi_0 , u_n \rangle_{L^2}}{\| \phi_0 \|^2_{L^2}} = \frac {\langle \phi_0 , v_{n-1} \rangle_{L^2} }{\| \phi_0 \|^2_{L^2}} = {\langle \phi_0 , v_{n-1} \rangle_{L^2} }\, , 
\ee
since $\phi_0 \perp u_n$ and $\| \phi_0 \|_{L^2}=1$. \ If we denote now
\be
\varphi_n = e_n \phi_0 + u_n - v_{n-1}
\ee
then $\varphi_n \perp \phi_0$ and thus we get
\be
\phi_n = [H-e_0 ]^{-1} \varphi_n \, . 
\ee
\end {proof}

\subsection {Main result} Here we state our main result.

\begin {theorem} \label {Teo1}
Let $d=1,2,3$ and let Hypotheses 1-2 be satisfied. \ Then, there exists $\nu^\star >0$ such that for any $\nu$ such that $|\nu |< \nu^\star$ the nonlinear Schr\"odinger equation (\ref {Eq1}) admits a stationary solution $\psi ( x, \nu) \in {\mathcal D}(H) \cap L^6 (\R^d )$, associated to an energy $E (\nu )$, given by means of the strong-convergent power series
\bee
\psi (x, \nu)= \sum_{n=0}^{\infty} \nu^n \phi_n (x) \ \mbox { and } \ E (\nu )= \sum_{n=0}^{\infty} \nu^n e_n \, , \label {quattro}
\eee
where $\phi_n (x)$ and $e_n$ are given in Lemma \ref {Lemma1}.
\end {theorem}

\begin {remark} \label {NotaD}
It is worth noting that the stationary solution $\psi$ given by (\ref {quattro}) is not normalized to one, that is, to the value of the norm of the unperturbed eigenvector $\phi_0$, which is assumed, for convenience of argument, to be equal to $1$. \ In fact, a simple calculation gives that
\be
\| \psi \|_{L^2}^2 = \sum_{n =0}^\infty \nu^n r_n \ \mbox { where } \ r_n = \sum_{m=0}^n \langle \phi_{n-m}, \phi_m \rangle \, . 
\ee
In particular
\be
r_0 = \| \phi_0 \|_{L^2}^2 =1 \, 
\ee
\be
r_1 = \langle \phi_0, \phi_1 \rangle + \langle \phi_1, \phi_0 \rangle =0 \, , \mbox { since } \ \phi_1 \in \Pi^\perp L^2 \, ,
\ee
and
\be
r_2 = \langle \phi_0, \phi_2 \rangle +\langle \phi_1, \phi_1 \rangle+\langle \phi_2, \phi_0 \rangle = \| \phi_1 \|_{L^2}^2 >0 \, ,  \mbox { since } \, \  \phi_2 \in \Pi^\perp L^2 \, .
\ee
Thus 
\be
\| \psi \|^2_{L^2} = 1 + \nu^2 g(\nu ) 
\ee
where $g(\nu ) $ is the analytic function obtained by means of the perturbative procedure for $\nu$ in a neighborhood of $\nu =0$ and such that $g(0) >0$. \ If one looks for a normalized solution may act as follows. \ Let
\be
\tilde \psi = \frac {\psi}{\| \psi \|_{L^2}}
\ee
be the normalized stationary solution to the equation
\be
H \tilde \psi + \tilde \nu |\tilde \psi |^2 \tilde \psi = E \tilde \psi
\ee
where $E$ is still given by (\ref {quattro}) and where
\bee
\tilde \nu := \tilde \nu (\nu )= \nu \| \psi \|^2_{L^2} = \nu \left [1 + \nu^2 g(\nu ) \right ] \, . \label {cinque}
\eee
Such a relation is invertible with inverse function
\be
\nu := \nu (\tilde \nu )\, .
\ee
In conclusion, if one look for the normalized solution to the equation
\bee
H \psi + \nu | \psi |^2  \psi = E  \psi \label {sei}
\eee
for a given value of the parameter $\nu$ let $\nu^\star$ be such that $\tilde \nu (\nu^\star )= \nu$, let $\psi$ and $E$ be the perturbative solutions given by (\ref {quattro}) corresponding to such a value of $\nu^\star$; then $\psi /\| \psi \|_{L^2}$ and $E$ are the normalized solution to (\ref {sei}).

In addition, by means of the scaling $\psi = \nu^{-2} \omega$ then (\ref {Eq1}) takes the form of the $\nu$-normalized equation 
\bee
H\omega +|\omega |^2 \omega = E \omega \label {EqII}
\eee
where we have just seen that
\be
\| \omega \|_{L^2}^2 = \nu \| \psi \|_{L^2}^2 = \nu \left [1+\nu g(\nu )\right ] \, . 
\ee
Thus, for $\nu$ in a neighborhood of $0$, we can find a continuous curve $(E(\nu ), \| \omega \|_{L^2} )$, near the point $(e_0 , 0)$, for solution to (\ref {EqII}). \  Recall that the analysis of the slope of this curve is important in the stability analysis of the
stationary state (see, e.g., the "slope condition" in \cite {Gr}).
\end {remark}

\section {$L^p$ estimates} \label {Sec5}

As anticipated in Remark \ref {Rem1} it turns out that formulas (\ref {Eq6}) and (\ref {Eq7}) make sense provided that $v_n$ and $\varphi_n$ belongs to $L^2$. \ Hence, we have to prove that $\phi_n$ belongs to $L^2 \cap L^6$ for any $n$. \ In order to obtain a $L^p$-norm estimate of the vectors $\phi_n$ we make use of the Gagliardo-Nirenberg inequality (\ref {Eq25}). 

\begin {lemma} \label {Lemma2}
Let $V(x)$ be a potential bounded from below (\ref {Eq2}); let $p$ and $C_{p,d}$ as given in (\ref {Eq25}) and (\ref {Eq26}). \ Concerning the $H^1$ and $L^p$ norms of $\phi_n$ we have that 
\be
\| \phi_n \|_{H^1} \le \mu_1 \| \varphi_n \|_{L^2} 
\ee
and
\be
\| \phi_n \|_{L^p} \le \mu_2 (p,d) \| \varphi_n \|_{L^2}
\ee
for some constants
\bee
\mu_1 = \left [ \frac {1}{\Lambda^2} +  \frac {1}{\Lambda } + \frac {e_0 - \Gamma}{\Lambda^2} \right ]^{1/2} \ \mbox { and } \ \mu_2 (p,d) := C_{p,d}  \frac {1}{\Lambda^{1-\rho} } \left [ \frac {1}{\Lambda } + \frac {e_0 - \Gamma}{\Lambda^2} \right ]^{\rho/2} 
\label {Eq27}
\eee
independent of $n$.
\end {lemma}

\begin {proof} 
From Remark \ref {Rem2} we have (\ref {Eq9}). \ Then, we have to estimate
\be
\| \nabla \phi_n \|_{L^2} = \| \nabla [H-e_0]^{-1} \varphi_n \|_{L^2}\, .
\ee
Since $V(x)$ is bounded from below, $V \ge \Gamma$, then
\be
\| \nabla [H-e_0 ]^{-1} \varphi_n \|^2_{L^2} &=& \langle [H-e_0 ]^{-1} \varphi_n , - \Delta [H-e_0 ]^{-1} \varphi_n \rangle_{L^2} \\ 
&=& \langle [H-e_0 ]^{-1} \varphi_n ,  \varphi_n \rangle_{L^2} - \langle [H-e_0 ]^{-1} \varphi_n , (V-e_0) [H-e_0 ]^{-1} \varphi_n \rangle_{L^2} \\
& \le & \langle [H-e_0 ]^{-1} \varphi_n ,  \varphi_n \rangle_{L^2} + (e_0 - \Gamma) \|  [H-e_0 ]^{-1} \varphi_n \|^2_{L^2}
\ee
since $-\Delta = H-e_0 - (V-e_0)$. \ Hence 
\be
\| \nabla [H-e_0 ]^{-1} \varphi_n \|_{L^2} \le \left [ \frac {1}{\Lambda } + \frac {e_0 - \Gamma}{\Lambda^2} \right ]^{1/2} \| \varphi_n \|_{L^2}
\ee
Therefore, we can conclude that
\be
\| \phi_n \|_{L^p} &\le & C_{p,d} \| \nabla \phi_n \|^\rho_{L^2}   \|  \phi_n \|^{1-\rho}_{L^2}  \\
& \le & C_{p,d}  \frac {1}{\Lambda^{1-\rho} } \| \varphi_n \|_{L^2}^{1-\rho}  
\left [ \frac {1}{\Lambda } + \frac {e_0 - \Gamma}{\Lambda^2} \right ]^{\rho/2} \| \varphi_n \|_{L^2}^\rho \\
& \le & \mu_2 \| \varphi_n \|_{L^2}
\ee
where $\mu_2 (p,d) $ is the constant (\ref {Eq27}) dependent on $p$ and $d$ but independent of $n$.
\end {proof}

\begin {lemma} \label {Lemma3}
Let $V(x)$ be a potential bounded from below: $V \ge \Gamma $. \ Let $d=1,2,3$ and let $\phi_j \in {\mathcal D}(H) \cap L^6$ for $j=0,\ldots , n-1$; then $\phi_n \in {\mathcal D} (H) \cap L^6$.
\end {lemma}

\begin {proof}
Indeed, $u_j \in L^2$, $j=0,\ldots , n$, by construction. \ Concerning $v_j$ we have that it belong to $L^2$ for any $j=0,\ldots , n-1$ from the H\"older inequality. \ Hence $E_n$ is well defined and $\varphi_j$ belongs to $L^2$ for any $j=1,\ldots , n$. \ From this fact and since $\varphi_j \perp \phi_0$,  $j=1,\ldots , n$, then $\phi_n = \left [ H- e_0 \right ]^{-1} \varphi_n \in {\mathcal D}(H)$.  Finally, by Lemma \ref {Lemma2} then $\phi_n \in L^6$ where we apply the Gagliardo-Nirenberg inequality (\ref {Eq25}) with  $p=6$.
\end {proof}

\begin {remark} \label {Rem5}
In fact, $\phi_j \in L^p$ for any $p \in [1,+\infty ]$ if $d=1$, $p \in [1,+\infty )$ if $d=2$ and $p \le 2d/(d-2)$ if $d>2$.
\end {remark}

\begin {remark} \label {Rem6}
From Lemma \ref {Lemma3} and from Remark \ref {Nota0} then we have proved that $\phi_n \in {\mathcal D}(H) \cap L^6$ for any $n=0,1,2 ,\ldots$, and $u_n$, $v_{n-1} \in L^2$ for any $n=1,2,\ldots $.
\end {remark}

\section {Are the formal series (\ref {Eq5}) convergent as $N$ goes to infinity? \ Proof of Theorem \ref {Teo1}}  \label {Sec6}

In order to prove the convergence of the perturbation series we give the following results.
\begin {lemma} \label {Lemma4}
Let 
\be
c_n := \| \phi_n \|_{L^2}\, , \ d_n := \| \phi_n \|_{L^6} \ \mbox { and } \ b_n := |e_n|\, ,\ n=0,1,2,\ldots \, , 
\ee
then {\bf for any $n \ge 1$ it follows that}
\be
b_n &\le & \sum_{m=0}^{n-1} d_m \sum_{\ell =0}^{n-1-m} d_\ell  d_{n-1-\ell-m} \\ 
c_n 
& \le & \frac {1}{\Lambda} \left [ 2 \sum_{m=0}^{n-1} d_m \sum_{\ell =0}^{n-1-m} d_\ell  d_{n-1-\ell-m} + \sum_{m=1}^{n-1} b_m c_{n-m} \right ] \\
d_n & \le & \mu_2 (6,d) \left [ 2 \sum_{m=0}^{n-1} d_m \sum_{\ell =0}^{n-1-m} d_\ell  d_{n-1-\ell-m} + \sum_{m=1}^{n-1} b_m c_{n-m} \right ]
\ee
\end {lemma}

\begin {proof} In order to prove the result above we remark that
\be
b_n = |e_n | \le \| v_{n-1} \|_{L^2} \, , \ 
c_n = \| \phi_n \|_{L^2} \le  \frac {1}{\Lambda}\| \varphi_n \|_{L^2}  
\ee
and 
\be 
d_n = \| \phi_n \|_{L^6} \le  \mu_2 (6,d) \| \varphi_n \|_{L^2} \, ,
\ee
from Lemma \ref {Lemma2}, where
\bee
\| \varphi_n \|_{L^2} \le  |e_n|  + \| v_{n-1} \|_{L^2} + \| u_n \|_{L^2} \le 2 \| v_{n-1} \|_{L^2} + \| u_n \|_{L^2} \, . \label {tre}
\eee
Hence, the above result follows since
\bee
\| v_{n-1}\|_{L^2} &\le & \sum_{m=0}^{n-1} \sum_{\ell =0}^{n-1-m} \| \phi_m \|_{L^6} \| \phi_\ell \|_{L^6} \| \phi_{n-1-\ell-m} \|_{L^6} \nonumber \\ 
& = & \sum_{m=0}^{n-1} \sum_{\ell =0}^{n-1-m} d_\ell d_m d_{n-1-\ell-m} \label {uno}
\eee
and
\bee 
\| u_n\|_{L^2}  \le  \sum_{m=1}^{n-1} b_m c_{n-m} \, . \label {due}
\eee
\end {proof}

\begin {lemma} \label {Lemma5}
Let us assume that
\bee
d_j \le \delta e^{\alpha j} \frac {1}{(j+1)^2} \, , \ j=0, \ldots , n-1 \, , \label {Eq28}
\eee
for some $\alpha >0$ and where
\be
\delta = d_0 = \| \phi_0 \|_{L^6} 
\, . 
\ee 
Then 
\bee
  \sum_{m=0}^{n-1} d_m \sum_{\ell =0}^{n-1-m} d_\ell  d_{n-1-\ell-m} \le C_1 \delta^3 e^{\alpha (n-1)} \frac {1}{(n+1)^2}
 \eee
 for any $n \ge 1$ and some $C_1 \le 4\cdot 4.7^2$.
\end {lemma}

\begin {remark}
By construction, (\ref {Eq28}) holds true for $j=0$.
\end {remark}

\begin {proof}
Indeed, from (\ref {Eq28}) it turns out that 
\be
 \sum_{m=0}^{n-1} d_m \sum_{\ell =0}^{n-1-m} d_\ell  d_{n-1-\ell-m} \le \delta^3 e^{\alpha (n-1)} I 
 \ee
 where we set
 \bee
 I:= \sum_{m=0}^{n-1} \frac {1}{(m+1)^2} \sum_{\ell =0}^{n-1-m} 
\frac {1}{(\ell +1)^2 (n-m-\ell )^2}\, . \label {Eq30}
\eee
A simple estimate proves that 
\be
\sum_{\ell =0}^{n-1-m} \frac {1}{(\ell +1)^2 (n-m-\ell )^2}  =  \frac {2}{(n-m)^2} + J(n-1-m)  \le \frac {4.7}{(n-m)^2}
\ee
where $J(n)$ has been defined and estimated in Appendix \ref {AppA}. \ Therefore, 
\be
I  &\le & \sum_{m=0}^{n-1} \frac {4.7}{(n-m)^2 (m+1)^2} \\
&\le & \frac {2\cdot 4.7}{n^2} + 4.7 J(n-1)  \le \frac {4.7^2}{n^2} \le \frac {4.7^2}{(n+1)^2}\frac {(n+1)^2}{n^2} \le \frac {4 \cdot 4.7^2}{(n+1)^2}
\ee
from which the statement follows.
\end {proof}

\begin {remark} \label {NotaA}
From Lemma \ref {Lemma5} and from (\ref {uno}) it follows that
\be
\| v_{n-1} \|_{L^2} \le C_1 \delta^3 e^{\alpha (n-1)} \frac {1}{(n+1)^2} \, .
\ee
\end {remark}

\begin {lemma} \label {Lemma6}
Let us assume that
\bee
b_j \le \beta e^{\alpha (j-1)} \frac {1}{(j+1)^2} \, , \ j=1, \ldots , n-1 \, , \label {Eq31}
\eee
and
\bee
c_j \le \gamma e^{\alpha j} \frac {1}{(j+1)^2} \, , \ j=1, \ldots , n-1 \, , \label {Eq32}
\eee
for some $\beta \ge 4 b_1 = 4 |e_1|$ and where
\be
\gamma =\max [ 4c_1 ,1] \, , \ c_1 = \| \phi_1 \|_{L^2}
\ee
and where $\alpha >0$ has been introduced in Lemma \ref {Lemma5}. \ Then 
\bee
 \sum_{m=1}^{n-1} b_m c_{n-m}  \le C_2 \beta \gamma  e^{\alpha (n-1)} \frac {1}{(n+1)^2}\, ,
 \eee
 for some $C_2 \le 2.7$.
 \end {lemma}
 
\begin {remark}
By construction, (\ref {Eq31}) and (\ref {Eq32}) hold true for $j=1$.
\end {remark}
 
 \begin {proof}
The proof immediately follows since
\be
\sum_{m=1}^{n-1} b_m c_{n-m} \le \beta \gamma e^{\alpha (n-1)} J (n)
\ee
where $J(n) \le \frac {2.7}{(n+1)^2}$ (see Appendix \ref {AppA}). \  \end {proof}

\begin {remark} \label {Rem7}
In fact, the estimate of the constants $C_1$ and $C_2$ are far to be optimal. \ Numerical analysis suggests that a sharp estimate for the term $I$ defined in (\ref {Eq30}) has the form
\be
I = \frac {g(n)}{(n+1)^2} \ \mbox { where } \ g(n) \le g(10) = 10.44589874\, ,
\ee
that is
\be
C_1 \le 10.45\, . 
\ee
Concerning $C_2$ from Appendix \ref {AppA} numerical analysis proves that
\be
C_2 \le 1.52 \, .
\ee
\end {remark}

\begin {remark} \label {NotaB}
From Lemma \ref {Lemma6} and from (\ref {due}) it follows that 
\be
\| u_n \|_{L^2} \le C_2 \beta \gamma e^{\alpha (n-1)} \frac {1}{(n+1)^2} \, . 
\ee
\end {remark}

\begin {remark} \label {NotaC}
From Remarks \ref {NotaA} and \ref {NotaB} and from (\ref {tre}) it follows that
\be
\| \varphi_n \|_{L^2} \le C_3 e^{\alpha (n-1)} \frac {1}{(n+1)^2} \, , 
\ee
where
\be
C_3  = \left [ 2 C_1 \delta^3 + C_2 \beta \gamma \right ]\, .
\ee
\end {remark}

Collecting Lemma \ref {Lemma4}, Lemma \ref {Lemma5} and Lemma \ref {Lemma6}, we have that
\bee
b_n & \le & C_1 \delta^3 e^{\alpha (n-1)} \frac {1}{(n+1)^2}\label {Eq34} \\ 
c_n 
 &\le & 
 \gamma \frac {e^{\alpha n}}{(n+1)^2} \frac {e^{-\alpha}}{\gamma \Lambda}C_3
 \label {Eq35} \\
d_n 
&\le & \delta \frac { e^{\alpha n}}{(n+1)^2} \frac { {\mu_2 (6,d) e^{-\alpha}} C_3 }{\delta } \label {Eq36} 
\eee
if (\ref {Eq28}), (\ref {Eq31}) and (\ref {Eq32}) hold true. 

In particular, if we choose 
\be
\beta = \max \left [ 4 b_1, C_1 \delta^3 \right ] 
\ee
%
and $\alpha >0$ large enough such that
\bee
\frac {e^{-\alpha}}{\gamma \Lambda}C_3 \le 1 \ \mbox { and } \  \frac { {\mu_2 (6,d) e^{-\alpha}} C_3 }{\delta }\le 1 \label {Eq24Bis}
\eee
then we have that (\ref {Eq28}), (\ref {Eq31}) and (\ref {Eq32}) hold true for $j=n$, too.

In conclusion, we have proved that
\begin {lemma} \label {Lemma7}
There exists four positive constants $\alpha >0$ large enough, $\beta >0$, $\gamma >0$ and $\delta >0$ independent of $n$ such that the following estimates 
\be
b_n \le \beta e^{\alpha (n-1)} \frac {1}{(n+1)^2} \, , \\ 
c_n \le \gamma e^{\alpha n} \frac {1}{(n+1)^2} \, ,  \\
d_n \le \delta  e^{\alpha n} \frac {1}{(n+1)^2} \, , 
\ee
hold true for any $n=1,2,\ldots $.
\end {lemma}

\begin {remark} \label {Rem8}
From Remark \ref {NotaC} and from Lemma \ref {Lemma2} it follows that $\psi_N$ is norm convergent in $H^1$.
\end {remark}

Finally:

\begin {theorem} \label {Teo2}
Let $d=1,2,3$ and let $\nu$ be such that $|\nu | < e^{-\alpha}$ where $\alpha >0$ is large enough as given in Lemma \ref {Lemma7}. \ Then the power series $E_N$ is absolutely convergent, and the power series $\psi_N$ is norm convergent in $L^2$ and $L^6$, and the power series $H\psi_N = \sum_{n=0}^{N} \nu^n H \phi_n$ is norm convergent in $L^2$
.
\end {theorem}

\begin {proof}
Convergence of $E_N$ and $\psi_N$ directly comes from Lemma \ref {Lemma7}. \ Concerning the convergence of $\sum_{n=0}^{N} \nu^n H \phi_n$ we simply remark that 
\be
\| H \phi_n \|_{L^2} &=& \| H [H-e_0 ]^{-1} \varphi_n \|_{L^2} \le \| \varphi_n \|_{L^2} + |e_0| \, \| [H-e_0]^{-1} \varphi_n \|_{L^2} \\ 
&\le & \left (1+|e_0| \Lambda^{-1} \right ) \| \varphi_n \|_{L^2} \le C_4 e^{\alpha n} \frac {1}{(n+1)^2}
\ee
for some $C_4 >0$, and thus the formal power series 
\be
\sum_{n} \nu^n H \phi_n 
\ee
is norm convergent in the space $L^2$ if $|\nu | <e^{-\alpha}$.
\end {proof}

So far we have proved that there exists vectors $u \, , \ w\, , \ \varphi \in L^2$, $v \in L^6$ and $z \in H^1$ such that 
\be
\psi_N \to u \ \mbox { in } \ L^2 \\ 
\psi_N \to v \ \mbox { in } \ L^6 \\ 
H \psi_N \to w \ \mbox { in } \ L^2 \\ 
\psi_N \to z \ \mbox { in } \ H^1 \\
\sum_{n=1}^N \nu^n \varphi_n \to \varphi \ \mbox { in } \ L^2
\ee
as $N$ goes to $\infty$. \ First of all we remark that 
\be
u &=& \sum_{n=0}^{\infty} \nu^n \phi_n = \phi_0 + \sum_{n=1}^{\infty} \nu^n \phi_n = \phi_0 + \sum_{n=1}^{\infty} \nu^n [H-e_0]^{-1} \varphi_n \\
&=& \phi_0 + [H-e_0]^{-1} \varphi
\ee
since $[H-e_0]^{-1}$ is a bounded operator on the eigenspace orthogonal to $\phi_0$, and where the convergence of the infinite sum has to be intended in the space $L^2$. \ Hence
\be
u \in {\mathcal D}(H) \, .  
\ee
Furthermore, we immediately have that
\be
\| u-z \|_{L^2} &\le & \| u - \psi_N \|_{L^2} + \| z - \psi_N \|_{L^2} \\ 
&\le & \| u - \psi_N \|_{L^2} + \| z - \psi_N \|_{H^1} \to 0 \, ,
\ee
hence $u=z$. \ Similarly, from (\ref {Eq25}) for some $\rho \in [0,1]$ we have that
\be
\| v-z \|_{L^6} &\le & \| v - \psi_N \|_{L^6} + \| z - \psi_N \|_{L^6} \\ 
&\le & \| v - \psi_N \|_{L^6} + C_{6,d} \| z - \psi_N \|_{H^1}^\rho \| z - \psi_N \|_{L^2}^{(1-\rho)} \\ 
&\le & \| v - \psi_N \|_{L^6} + \| z - \psi_N \|_{H^1}^\rho \left [ \| z \|_{L^2} + \| \psi_N \|_{L^2} \right ]^{(1-\rho)}
\to 0 \, ,
\ee
hence $v=z$. \ In conclusion, there exists a vector $\psi \in {\mathcal D} (H)$ such that 
\be
\psi_N \to \psi \ \mbox { in } \ L^2\, , \ L^6 \ \mbox { and } \ H^1\, , 
\ee
and
\be
H \psi & = & [H-e_0] \psi + e_0 \psi = [H-e_0] \left ( \phi_0 + [H-e_0]^{-1} \varphi \right ) + e_0 \psi \\
&=& e_0 \psi + \varphi \\ 
H \psi_N & = & H \sum_{n=0}^N \nu^n \phi_n = H \phi_0 + \sum_{n=1}^N \nu^n H \phi_n = H \phi_0 + \sum_{n=1}^N \nu^n  H [H-e_0]^{-1} \varphi_n \\
&=& e_0 \sum_{n=0}^N \nu^n \phi_n + \sum_{n=1}^N \nu^n  \varphi_n \to e_0 \psi + \varphi = H \psi
\ee
as $N\to \infty$.

Thus we have proved the following result.

\begin {lemma} \label {Lemma8}
$\psi_N \to \psi \in {\mathcal D} (H)$ in $L^2$, $L^6$ and $ H^1$, and $H\psi_N \to H\psi$ in $L^2$. 
\end {lemma}

Finally, it's not hard to see that $\psi$ is a stationary solution associated to the energy $E$ to (\ref {Eq1}). \ Indeed:

\begin {lemma} \label {Lemma9}
Let 
\be
r_N := H \psi_N + \nu |\psi_N |^2 \psi_N - E_N \psi_N 
\ee
then 
\bee
\lim_{N\to \infty} \| r_N \|_{L^2} =0 \, . \label {Eq38} 
\eee
\end {lemma}

\begin {proof}
A simple straightforward calculation gives that 
\be
r_N &=& \sum_{n=0}^N \nu^n H \phi_n + \sum_{n,m,\ell =0}^N \nu^{n+m+\ell +1} \phi_n \phi_m \bar \phi_\ell - \sum_{n,m=0}^N \nu^{n+m} e_m \phi_n \\ 
&=& \sum_{n=0}^N \nu^n H \phi_n + \sum_{n=1}^{3N+1} \nu^n \left [ \sum_{m=0}^{n-1} \sum_{\ell =0}^{n-1-m} \phi_m \bar \phi_{\ell} \phi_{n-m-\ell -1} \right ] - \sum_{n=0}^{2N} \nu^n \left [ \sum_{m=0}^n e_m \phi_{n-m} \right ] \\ 
&=& \sum_{n=N+1}^{3N+1} \nu^n v_{n-1} - \sum_{n=N+1}^{2N} \nu^n u_n
\ee
where the two power series $\sum_{n=0}^{\infty} \nu^n v_{n-1}$ and $\sum_{n=1}^{\infty} \nu^n u_{n}$ are norm-$L^2$ convergent for $\nu$ small enough. \ Then (\ref {Eq38}) follows.
\end {proof}

Now, we are ready to complete the proof of Theorem \ref {Teo1}. 

Indeed, let 
\be
r &:=& H \psi + \nu |\psi |^2 \psi - E \psi = a_N + b_N + c_N + r_N
\ee
where
\be
a_N &=& H \psi - H \psi_N \to 0 \ \mbox { in } L^2 \\ 
b_N &=& \nu \left [ |\psi |^2 \psi - |\psi_N |^2 \psi_N \right ]  \to 0 \ \mbox { in } L^2 \\ 
c_N &=& - \left [ E \psi - E_N \psi_N \right ]  \to 0 \ \mbox { in } L^2 \
\ee
From the above results immediately follows that 
\be
\| a_N\|_{L^2} \to 0 \ \mbox { and } \ \| c_N\|_{L^2} \to 0  
\ee
as $N$ goes to infinity. \ Concerning $b_N$ one notes that
\be
\| b_N \|_{L^2} &\le & \| (\psi - \psi_N ) |\psi |^2 \|_{L^2} + \| (\psi - \psi_N ) \bar \psi \psi_N \|_{L^2}  + \| (\bar \psi - \bar \psi_N ) |\psi_N |^2 \|_{L^2} \\
&\le & \| |\psi |^2 \|_{L^3} \| \psi - \psi_N \|_{L^6} + \| \psi - \psi_N \|_{L^6} \| \bar \psi \psi_N \|_{L^3} +  \| |\psi_N |^2 \|_{L^3} \| \psi - \psi_N \|_{L^6} \\
& \le & \left [ \| |\psi |^2 \|_{L^3} + \| \bar \psi \psi_N \|_{L^3} +  \| |\psi_N |^2 \|_{L^3} \right ] \| \psi - \psi_N \|_{L^6} \to 0 
\ee
as $N$ goes to infinity. \ From these facts and since (\ref {Eq38}) then Theorem \ref {Teo1} is proved. 

\begin {remark} \label {Nota}
Since the constants $\beta $, $\gamma$, $\delta$, $\Lambda$, $C_1$, $C_2$, $C_3$, $C_4$ and $\mu_2 (6,d)$ can be estimated then one can obtain the value of the parameter $\alpha$ solution to (\ref {Eq24Bis}). \ Hence, the estimate $\nu^\star < e^{-\alpha}$ of the radius of convergence follows.
\end {remark}

\section {A Toy model - infinite well potential} \label {Sec3} 
Let us consider, in dimension one, the infinite well potential of the form:
\be
V(x) = 
\left \{
\begin {array}{l} 
0 \ \mbox { if } |x| < \pi \\ 
+\infty \mbox { if } |x| \ge \pi 
\end {array}
\right. \, . 
\ee

\subsection {Linear time-independent Schr\"odinger Equation} The linear operator $H$ is formally defined as follows:
\be
H \psi = - \psi'' \, ,\ x \in (-\pi , + \pi ) \, , \ \psi \in L^2 ((-\pi , +\pi )) \, 
\ee
with Dirichlet boundary conditions
\bee
\psi (- \pi )= \psi (+\pi )=0  \, .\label {Eq12}
\eee
By means of a straightforward calculation it follows that the spectrum of $H$ is purely discrete and it is given by means of simple eigenvalues
\be
\lambda_j = \frac 14 j^2 \, , \ j=1,2, \ldots \, ,
\ee
with associated normalized eigenvectors
\be
q_j (x) = \frac {1}{\sqrt {\pi}}
\left \{
\begin {array}{l} 
\cos (j x/2) \, , \  \mbox { odd } j \\ 
\sin (j x/2) \, , \ \mbox { even } j
\end {array}
\right. \, . 
\ee
The resolvent operator is given by
\bee
\left ( [H-z ]^{-1} \psi \right ) (x) = \sum_{j=1}^\infty \frac {1}{\lambda_j - z} q_j (x) \langle q_j , \psi \rangle_{L_2} \, . \label {Eq13}
\eee

\subsection {Perturbation theory}

By making use of the perturbation formula we compute now the coefficients of the formal power series (\ref {Eq5}) where 
\be
e_0 = \lambda_1 = \frac 14
\ee
is the first unperturbed eigenvalue with associated unperturbed eigenvector 
\be
\phi_0 (x)= q_1 (x) = \frac {1}{\sqrt {\pi}} \cos (x/2)\, .
\ee

\begin {remark} \label {Rem3}
Here, we have considered, for argument's sake, the formal power series (\ref {Eq5}) associated to the first eigenvalue $\lambda_1$. \ Similarly, the same method may be applied to the unperturbed eigenvalues $\lambda_j$ for any $j>1$.
\end {remark}

The perturbation theory exploited in Lemma \ref {Lemma1} gives that
\be
v_0 = \phi_0^3 \, , \ u_1 =0 \, , \ 
e_1 = \frac {\| \phi_0^2 \|^2}{\| \phi_0 \|^2}= \frac {3}{4\pi}
\ee
and 
\be
\varphi_1 = e_1 \phi_0 - \phi_0^3 \, . 
\ee
Finally
\be
\phi_1 &=& 
\left ( [H-e_0 ]^{-1} \varphi_1 \right ) (x) = \sum_{j=2}^\infty \frac {1}{\lambda_j - e_0} q_j (x) \langle q_j , \varphi_1 \rangle_{L_2} \\ 
&=& -\frac {1}{\lambda_3 - e_0} q_3 (x) \langle q_3 , \phi_0^3 \rangle_{L_2} = -\frac {1}{8[\pi]^{3/2}} \cos \left ( \frac {3}{2} x \right )
\ee

By means of a straightforward calculation the other terms follow; for instance
\be
e_2 = -\frac {3}{32 \pi^2} \, ,\ e_3 = \frac {15}{256 \pi^3} \, , \ e_4 = -\frac {69}{2048 \pi^4} \\ 
e_5 = \frac {75}{4096 \pi^5} \, , \ e_6 = -\frac {1257}{131072 \pi^6} \, , \ e_7 = \ldots 
\ee
and
\be
\phi_2 (x) &=&  \frac {1}{64{\pi}^{5/2}} \left [ 3\cos \left ( \frac {3x}{2} \right ) +  \cos \left ( \frac {5x}{2} \right ) \right ] \\
\phi_3 (x) &=& - \frac {1}{512{\pi}^{7/2}} \left [ 9\cos \left ( \frac {3x}{2} \right ) +5  \cos \left ( \frac {5x}{2} \right ) +  
\cos \left ( \frac {7x}{2} \right ) \right ] \\
\phi_4 (x) &=& \frac {1}{4906{\pi}^{9/2}} \left [ 27\cos \left ( \frac {3x}{2} \right ) +20  \cos \left ( \frac {5x}{2} \right ) +
7  \cos \left ( \frac {7x}{2} \right ) +  \cos \left ( \frac {9x}{2} \right ) \right ] \\
\phi_5 (x) &=& - \frac {1}{32768{\pi}^{11/2}} \left [ 81\cos \left ( \frac {3x}{2} \right ) +75  \cos \left ( \frac {5x}{2} \right ) +
35  \cos \left ( \frac {7x}{2} \right )  + \right. \\ 
&& \left. \ \  + 9 \cos \left ( \frac {9x}{2} \right ) + \cos \left ( \frac {11x}{2} \right ) \right ] \\
\phi_6 (x) &=& \frac {1}{262144{\pi}^{13/2}} \left [ 243\cos \left ( \frac {3x}{2} \right ) +275  \cos \left ( \frac {5x}{2} \right ) +
154  \cos \left ( \frac {7x}{2} \right ) + \right. \\ 
&& \left. \ \ + 54 \cos \left ( \frac {9x}{2} \right ) +  11 \cos \left ( \frac {11x}{2} \right ) + \cos \left ( \frac {13x}{2} \right ) \right ] \\
\phi_7 (x) &=& \ldots 
\ee

\begin{table}
\begin{center}
\begin{tabular}{|c||c|c|c||c|c|c|} 
\hline
\multicolumn{1}{|c||}{} & \multicolumn{3}{|c||}{$\nu =0.1$} & \multicolumn{3}{|c|}{$\nu =1$}
\\ \hline
$N $    &  $E_N$          & $\| \psi_N \|_{L^2} $ &  $\| r_N \|_{L^2} $    &  $E_N$         & $\| \psi_N \|_{L^2} $  & $\| r_N \|_{L^2} $    \\ \hline \hline 
$0 $    &  $ 0.25 $       & $1$                   & $0.25 \cdot 10^{-1}$   &  $ 0.25 $       & $1$                    & $0.25 \cdot 10^{0}$    \\ \hline
$1 $    &  $ 0.273873 $   & $ 1.000007916 $                   & $0.16 \cdot 10^{-3}$   &  $ 0.488732 $   & $ 1.000791259 $                    & $0.16 \cdot 10^{-1}$   \\ \hline
$2 $    &  $ 0.273778 $   & $ 1.000007728 $                   & $0.30 \cdot 10^{-5}$   &  $ 0.479234 $   & $ 1.000614941 $                    & $0.28 \cdot 10^{-2}$   \\ \hline
$3 $    &  $ 0.273780 $   & $ 1.000007730 $                   & $0.52 \cdot 10^{-7}$   &  $ 0.481123 $   & $ 1.000634507 $                    & $0.48 \cdot 10^{-3}$   \\ \hline
$4 $    &  $ 0.273780 $   & $ 1.000007730 $                   & $0.88 \cdot 10^{-9}$   &  $ 0.480777 $   & $ 1.000632165 $                    & $0.81 \cdot 10^{-4}$   \\ \hline
$5 $    &  $ 0.273780 $   & $ 1.000007730 $                   & $0.15 \cdot 10^{-10}$  &  $ 0.480837 $   & $ 1.000632442 $                    & $0.13 \cdot 10^{-4}$   \\ \hline
$6 $    &  $ 0.273780 $   & $ 1.000007730 $                   & $0.24 \cdot 10^{-12}$  &  $ 0.480827 $   & $ 1.000632410 $                    & $0.22 \cdot 10^{-5}$   \\ \hline
\end{tabular}
\caption{Infinite well potential - Table of values corresponding to the case of defocusing nonlinearities when $\nu =0.1$ and $\nu=+1$.}
\label{tab1}
\end{center}
\end {table}
\begin{table}
\begin{center}
\begin{tabular}{|c||c|c|c||c|c|c|} 
\hline
\multicolumn{1}{|c||}{} & \multicolumn{3}{|c||}{$\nu =-0.1$} & \multicolumn{3}{|c|}{$\nu =-1$}
\\ \hline
$N $    &  $E_N$          & $\| \psi_N \|_{L^2} $ &  $\| r_N \|_{L^2} $    &  $E_N$         & $\| \psi_N \|_{L^2} $  & $\| r_N \|_{L^2} $    \\ \hline \hline 
$0 $    &  $ 0.25 $       & $1$                   & $0.25 \cdot 10^{-1}$   &  $ 0.25 $       & $1$                    & $0.25 \cdot 10^{0}$    \\ \hline
$1 $    &  $ 0.226168 $   & $ 1.000007916 $                   & $0.17 \cdot 10^{-3}$   &  $ 0.011268 $   & $ 1.000791259 $                    & $0.17 \cdot 10^{-1}$   \\ \hline
$2 $    &  $ 0.226032 $   & $ 1.000008160 $                   & $0.30 \cdot 10^{-5}$   &  $ 0.001769 $   & $ 1.000992585 $                    & $0.32 \cdot 10^{-2}$   \\ \hline
$3 $    &  $ 0.226030 $   & $ 1.000008108 $                   & $0.53 \cdot 10^{-7}$   &  $ -0.000121 $   & $ 1.001018592 $                    & $0.58 \cdot 10^{-3}$   \\ \hline
$4 $    &  $ 0.226030 $   & $ 1.000008108 $                   & $0.90 \cdot 10^{-9}$   &  $ -0.000467 $   & $ 1.001021668 $                    & $1.00 \cdot 10^{-4}$   \\ \hline
$5 $    &  $ 0.226030 $   & $ 1.000008108 $                   & $0.15 \cdot 10^{-10}$  &  $ -0.000527 $   & $ 1.001022048 $                    & $0.17 \cdot 10^{-4}$   \\ \hline
$6 $    &  $ 0.226030 $   & $ 1.000008108 $                   & $0.25 \cdot 10^{-12}$  &  $ -0.000537 $   & $ 1.001022093 $                    & $0.28 \cdot 10^{-5}$   \\ \hline
\end{tabular}
\caption{Infinite well potential - Table of values corresponding to the case of focusing nonlinearities when $\nu =-0.1$ and $\nu=-1$.}
\label{tab2}
\end{center}
\end {table}

In Tables \ref {tab1}-\ref {tab2} we compute the values of $E_N$, $\| \psi_N \|_{L^2}$ and 
$\| r_N \|_{L^2}$, where $r_N$ is the remainder term defined by (\ref {Eq8}), for different values of $N$ and for $\nu =\pm 0.1$ and $\nu = \pm1$. \ It turns out that the formal power series (\ref {Eq5}) rapidly converges and that the norm of the remainder term $r_N$ rapidly decreases when $N$ increases.
                               
\subsection {Nonlinear time-independent Schr\"odinger Equation}

The nonlinear time-independent equation (\ref {Eq1}) becomes
\bee
- \psi'' + \nu |\psi |^2 \psi = E \psi \, , \ x \in (-\pi , +\pi )\, , \label {Eq14}
\eee
with Dirichlet boundary conditions (\ref  {Eq12}). \ 
If we restrict our attention to the case of $\nu \in \R$ and $E\in \R$ then we known that the stationary solution is, up to a constant phase factor, a real-valued 
function. \ The proof of this result is quite similar to the one of Lemma 3.7 given by \cite {P}. \ Indeed, if we multiply both sides of (\ref {Eq14}) by $\bar \psi$, we obtain 
that
\be
- \psi'' \bar \psi + \nu |\psi |^4 = E |\psi |^2
\ee
and similarly 
\be
- \bar \psi'' \psi + \nu |\psi |^4 = E |\psi |^2 \, . 
\ee
Hence
\be
0 = \psi'' \bar \psi - \bar \psi'' \psi = (\psi' \bar \psi - \bar \psi' \psi )' 
\ee
from wich follows that 
\be 
\psi' \bar \psi - \bar \psi' \psi = C \, , \forall x \in (-\pi , +\pi ) \, , 
\ee
for some constant $C$. \ Recalling that $\psi (\pm \pi )=0$ then $C=0$ and thus $\theta = \arg (\psi )$ is a constant term.

Therefore, stationary solutions $\psi$ to (\ref {Eq14}) may be assumed to be real-valued and they satisfy to the equation
\bee
- \psi'' + \nu \psi^3 = E \psi \, , \ x \in (-\pi , +\pi )\, . \label {Eq15}
\eee

The general solution to such an equation has the form \cite {D}
\be
\psi (x) = \chi \mbox {sn} \left [ \zeta (x-x_0) ,k \right ] \, , \ \chi \in \R \, , 
\ee
where $x_0$ and $\zeta$ are arbitrary constants and where
\be
k^2 = \frac {E-\zeta^2}{\zeta^2} \ \mbox { and } \ \chi^2 =2 \frac {E-\zeta^2}{\nu}\, . 
\ee
The Dirichlet boundary conditions imply that $x_0 =-\pi$ and that $2\zeta \pi $ is a zero of the Jacobian Elliptic function $\mbox {sn}(x,k)$, i.e.: 
\be
2 \zeta  \pi = 2 K(k) m, \ m=1,2,\ldots \, , 
\ee
where $K(k)$ and $E(k)$ are the complete elliptic integral of first and second kind. \ The norm of the wavefunction $\psi$ is given by
\be
\| \psi \|_{L^2}^2 &=& \chi^2 \int_{-\pi}^{+\pi} \mbox {sn}^2 \left [ \zeta (x+\pi) ,k \right ] dx = 2m\frac {\chi^2}{\zeta} \frac {K(k)-E(k)}{k^2} \\
&=& 2\pi \chi^2  \frac {K(k)-E(k)}{K(k) k^2}\, . 
\ee
Hence
\be
\chi^2 = \frac {1}{2\pi} \frac {K(k) k^2}{K(k)-E(k)} \| \psi \|_{L^2}^2 \, . 
\ee

In order to find the stationary solutions to (\ref {Eq15}) the quantization conditions read
\be
2\frac {E-\zeta^2}{\nu}= \frac {1}{2\pi} \frac {K(k) k^2}{K(k)-E(k)} \| \psi \|_{L^2}^2 \, , \ k^2 = \frac {E-\zeta^2}{\zeta^2} 
\ \mbox { and } \ \zeta=\frac {K(k)}{\pi} m \, , \ m=1,2,\ldots \, . 
\ee

\subsubsection {Defocusing nonlinearity: $\nu >0$} 
When $\nu > 0$ then stationary solutions there exist provided that $E-\zeta^2 \ge 0$ and $k$ is a real-valued solution to the equation
\bee 
 K(k) [ K(k)-E(k) ] = \frac {\nu \pi}{4 m^2} \| \psi \|_{L^2}^2 \label {Eq16}
\, . 
\eee
If we remark that the function $K(k) [ K(k)-E(k) ]$ is a monotone increasing function such that 
\be
\lim_{k\to 0^+} K(k) [ K(k)-E(k) ] =0 \ \mbox { and } \ 
\lim_{k\to 1^-} K(k) [ K(k)-E(k) ] =+\infty 
\ee
then the equation above (\ref {Eq16}) has a unique solution $k_m \in (0,1)$, for any $m=1,2,\ldots $ fixed, and then there exists a family of values of the parameter $E$: 
\bee
E 
=   \left [ \frac {K(k_m)m}{\pi} \right ]^2 \left [1 + k_m^{2} \right ]\, ,\ m=1,2,\ldots \, . \label {Eq17}
\eee

\subsubsection {Focusing nonlinearity: $\nu <0$} 

On the other hands if $\nu <0$ then stationary solutions there exist provided that $E-\zeta^2 \le 0$ and $k=i \kappa$, $\kappa \in \R$, is a purely imaginary complex number; in such a case we recall that 
\be
\mbox {sn} (x,i\kappa )=k_1' \mbox {sd} (x\sqrt {1+\kappa^2},k_1) \ \mbox { where } \ \mbox {sd} (x,k_1) = \frac {\mbox {sn} (x,k_1)}{\mbox {dn} (x,k_1)}
\ee
and where
\be
k_1 = \kappa /\sqrt{1+\kappa^2} \mbox { and } k_1'= \sqrt {1-k_1^2}= \frac {1}{\sqrt {1+\kappa^2}}\, . 
\ee
Hence, equation (\ref {Eq16}) becomes 
\bee 
 K(i\kappa ) [ K(i\kappa )-E(i\kappa ) ] = \frac {\nu \pi}{4 m^2} \| \psi \|_{L^2}^2
\, . \label {Eq18}
\eee
If we remark that the function $K(i\kappa ) [ K(i\kappa )-E(i\kappa ) ]$ is a monotone real-valued decreasing function for $\kappa \in [0,+\infty )$ such that 
\be
\lim_{\kappa \to 0^+} K(i\kappa ) [ K(i\kappa )-E(i\kappa ) ] =0 \ \mbox { and } \ 
\lim_{\kappa \to +\infty} K(i\kappa ) [ K(i\kappa )-E(i\kappa ) ] =-\infty 
\ee
then the equation above (\ref {Eq18}) has a unique solution $\kappa_m \in (0,+\infty )$ for any $m=1,2,\ldots $ fixed, and also in this case there exists a family of values of the parameter $E$:
\bee
E 
=   \left [ \frac {K(i \kappa_m)m}{\pi} \right ]^2 \left [1 - \kappa_m^{2} \right ]\, , \ m=1,2,\ldots \, . \label {Eq19}
\eee

\subsection {Comparison between the perturbative result and the exact one}

From Table \ref {tab1} the perturbative result gives that the stationary solution to (\ref {Eq1}) for $\nu =0.1$ and $N=6$ has energy
\bee
E_6 =0.273780 \label {Eq20}
\eee
with associated wavefunction with norm
\bee
\| \psi \|_{L^2} \approx \| \psi_6 \|_{L^2} = 1.000007730\, . \label {Eq21}
\eee
The value of the solution $k$ to (\ref {Eq16}), where $m=1$ and where the value of $\| \psi \|_{L^2}$ is the one in (\ref {Eq21}), is given by
\be
k = 0.2474031338
\ee
and the associated energy $E$ is given by (\ref {Eq17})
\be
E= 0.273780
\ee
in full agreement with (\ref {Eq20}). \ Similarly, For $\nu =1$ and $N=6$ then Table \ref {tab1} gives that 
\bee
E_6  = 0.480827 \label {Eq22}
\eee
with associated wavefunction with norm
\be
\| \psi \|_{L^2} \approx \| \psi_6 \|_{L^2} = 1.000632410 \, . 
\ee
From (\ref {Eq16}) and (\ref {Eq17}) then (where $m=1$)
\be
k = 0.6682383718 \ \mbox { and } \ 
E= 0.480829
\ee
in good agreement with (\ref {Eq22}).

From Table \ref {tab2} the perturbative result gives that the stationary solution for $\nu =-0.1$ and $N=6$ has energy
\bee
E_6 =0.226030 \label {Eq23}
\eee
with associated wavefunction with norm
\be
\| \psi \|_{L^2} \approx \| \psi_6 \|_{L^2} = 1.000008108\, . 
\ee
From (\ref {Eq18}) and (\ref {Eq19}) then (where $m=1$)
\be
\kappa = 0.2574471610 \ \mbox { and } \ 
E= 0.226030
\ee
in full agreement with (\ref {Eq23}). \ Similarly, For $\nu =-1$ and $N=6$ then Table \ref {tab2} gives that 
\bee
E_6  = -0.000537 \label {Eq24}
\eee
with associated wavefunction with norm
\be
\| \psi \|_{L^2} \approx \| \psi_6 \|_{L^2} = 1.001022099 \, . 
\ee
From (\ref {Eq18}) and (\ref {Eq19}) then (where $m=1$)
\be
\kappa = 1.001546564 \ \mbox { and } \ 
E= -0.000540
\ee
in very good agreement with (\ref {Eq24}).

\section {Harmonic oscillator} \label {Sec4}
Let us consider, in dimension one, the harmonic oscillator with potential
\be
V(x) =  x^2 \, , \ x \in \R \, . 
\ee
That is, the linear operator $H$ is defined as follows:
\be
H \psi = - \psi'' + x^2 \psi \, ,\ x \in \R \, , \ \psi \in L^2 (\R ) \,  . 
\ee
It is well known that the spectrum of $H$ is purely discrete and it is given by simple eigenvalues
\be
\lambda_j = 2j- 1 \, , \ j =1,2 , \ldots \, , 
\ee
with associated normalized eigenvectors
\be
q_j (x) =\frac {1}{\sqrt {2^{j-1}(j-1)!}} \left ( \frac {1}{\pi} \right )^{1/4} e^{-  x^2 /2} H_{j-1} \left (  x \right )  \, ,
\ee
where 
\be
H_j (x) = (-1)^j e^{x^2} \frac {d^j}{dx^j} e^{-x^2} 
\ee
are the Hermite's polynomials. 

By making use of the perturbation formula we compute now the coefficients of the formal power series (\ref {Eq5}) associated to the first unperturbed eigenvalue 
\be
e_0 = \lambda_1 = 1
\ee
with associated unperturbed normalized eigenvector 
\be
\phi_0 (x)= q_1 (x) = \frac {1}{\sqrt[4] {\pi}} e^{-x^2 /2} \, .
\ee
In such a case the perturbative procedure gives that
\be
e_1 = \frac {\| \phi_0^2 \|_{L^2}^2}{\| \phi_0 \|_{L^2}^2} = \frac {1}{\sqrt {2\pi }}\, . 
\ee
Furthermore, 
\be
\phi_1 &=& \left [ H- e_0 \right ]^{-1} \varphi_1 = \sum_{j=2}^{+\infty } \frac {1}{\lambda_j - e_0} q_j (x) \langle q_j , \varphi_1 \rangle_{L^2} \\
&\approx & \sum_{j=2}^{N_2} \frac {1}{\lambda_j - e_0} q_j (x) \langle q_j , \varphi_1 \rangle_{L^2} \, , 
\ee
where $\varphi_1 = e_1 \phi_0 - \phi_0^3 $ and where the resolvent operator is given by a with infinitely many terms. \ In numerical calculation we truncate the series for $j$ up to a some  large enough positive integer $N_2$; in numerical experiments we observe that $N_2 = 60$ is a suitable value. \ Iterating such a procedure we can 
obtain in Tables \ref {tab3}-\ref {tab4} the numerical values of $E_N$, $\| \psi_N \|_{L^2}$ and $\| r_N \|_{L^2}$ for $N=1,2,\ldots , 6$, where $r_N$ is the remainder term defined by (\ref {Eq8}), for $\nu =\pm 0.1$ and $\nu =\pm 1$. \ As in the toy model discussed in Section \ref {Sec3} it turns out that the formal power series seems to rapidly converges for $|\nu |\le 1$.

\begin{table}
\begin{center}
\begin{tabular}{|c||c|c|c||c|c|c|} 
\hline
\multicolumn{1}{|c||}{} & \multicolumn{3}{|c||}{$\nu =+0.1$} & \multicolumn{3}{|c|}{$\nu =+1$}
\\ \hline
$N $    &  $E_N$             & $\| \psi_N \|_{L^2} $ &  $\| r_N \|_{L^2} $    &  $E_N$             & $\| \psi_N \|_{L^2} $  & $\| r_N \|_{L^2} $    \\ \hline \hline 
$0 $    &  $1 $              & $1$                   & $0.43 \cdot 10^{-1}$   &  $ 1    $          & $1$                    & $0.43 \cdot 10^{0}$    \\ \hline
$1 $    &  $ 1.039894228 $   & $ 1.000006539 $       & $0.24 \cdot 10^{-3}$   &  $ 1.398942280 $   & $ 1.000653715 $        & $0.23 \cdot 10^{-1}$   \\ \hline
$2 $    &  $ 1.039728699 $   & $ 1.000006483 $       & $0.25 \cdot 10^{-5}$   &  $ 1.382389419 $   & $ 1.000599885 $        & $0.24 \cdot 10^{-2}$   \\ \hline
$3 $    &  $ 1.039730376 $   & $ 1.000006484 $       & $0.24 \cdot 10^{-7}$   &  $ 1.384066368 $   & $ 1.000602159 $        & $0.23 \cdot 10^{-3}$   \\ \hline
$4 $    &  $ 1.039730361 $   & $ 1.000006484 $       & $0.22\cdot 10^{-9}$    &  $ 1.383909162 $   & $ 1.000602036 $        & $0.21 \cdot 10^{-4}$   \\ \hline
$5 $    &  $ 1.039730361 $   & $ 1.000006484 $       & $0.21 \cdot 10^{-11}$  &  $ 1.383923548 $   & $ 1.000602043 $        & $0.18 \cdot 10^{-5}$   \\ \hline
$6 $    &  $ 1.039730361 $   & $ 1.000006484 $       & $0.43 \cdot 10^{-12}$  &  $ 1.383922248 $   & $ 1.000602043 $        & $0.17 \cdot 10^{-6}$   \\ \hline
\end{tabular}
\caption{Harmonic oscillator potential - table of values corresponding to the case of defocusing nonlinearities when $\nu =0.1$ and $\nu=+1$.}
\label{tab3}
\end{center}
\end {table}
\begin{table}
\begin{center}
\begin{tabular}{|c||c|c|c||c|c|c|} 
\hline
\multicolumn{1}{|c||}{} & \multicolumn{3}{|c||}{$\nu =-0.1$} & \multicolumn{3}{|c|}{$\nu =-1$}
\\ \hline
$N $    &  $E_N$             & $\| \psi_N \|_{L^2} $ &  $\| r_N \|_{L^2} $    &  $E_N$             & $\| \psi_N \|_{L^2} $  & $\| r_N \|_{L^2} $    \\ \hline \hline 
$0 $    &  $1 $              & $1$                   & $0.43 \cdot 10^{-1}$   &  $ 1    $          & $1$                    & $0.43 \cdot 10^{0}$    \\ \hline
$1 $    &  $ 0.9601057720 $   & $ 1.000006539 $       & $0.24 \cdot 10^{-3}$   &  $ 0.6010577196 $   & $ 1.0006537148 $        & $0.25 \cdot 10^{-1}$   \\ \hline
$2 $    &  $ 0.9599402433 $   & $ 1.000006596 $       & $0.25 \cdot 10^{-5}$   &  $ 0.5845048581 $   & $ 1.0007119732 $        & $0.26 \cdot 10^{-2}$   \\ \hline
$3 $    &  $ 0.9599385664 $   & $ 1.000006596 $       & $0.24 \cdot 10^{-7}$   &  $ 0.5828279087 $   & $ 1.0007147627 $        & $0.25 \cdot 10^{-3}$   \\ \hline
$4 $    &  $ 0.9599385507 $   & $ 1.000006596 $       & $0.22\cdot 10^{-9}$    &  $ 0.5826707021 $   & $ 1.0007149164 $        & $0.24 \cdot 10^{-4}$   \\ \hline
$5 $    &  $ 0.9599385505 $   & $ 1.000006596 $       & $0.20 \cdot 10^{-11}$  &  $ 0.5826563161 $   & $ 1.0007149254 $        & $0.22 \cdot 10^{-5}$   \\ \hline
$6 $    &  $ 0.9599385505 $   & $ 1.000006596 $       & $0.28 \cdot 10^{-12}$  &  $ 0.5826550168 $   & $ 1.0007149260 $        & $0.20 \cdot 10^{-6}$   \\ \hline
\end{tabular}
\caption{Harmonic oscillator potential - table of values corresponding to the case of focusing nonlinearities when $\nu =-0.1$ and $\nu=-1$.}
\label{tab4}
\end{center}
\end {table}

\section {Conclusions} \label {Sec7}

In this paper we have applied the Rayleigh-Schr\"odinger perturbation theory when the unperturbed linear operator has an isolated nondegenerate eigenvalue and where the nonlinear term plays the role of the perturbation. \ The power series has coefficients that can be iteratively obtained and such a series is proved to be convergent when the strength $\nu$ of the nonlinear term has absolute value less than a threshold value $\nu^\star$, for some $\nu^\star >0$. 

From the numerical experiments resumed in Tables \ref {tab1}, \ref {tab2} and Tables \ref {tab3}, \ref {tab4} one has evidence  that the formal power series (\ref {Eq5}) rapidly converges for $|\nu|$ small enough when $N$ goes to infinity. \ In particular, from Figure \ref {fig1} one can see that $|e_n|$ and $\| \phi_n \|_{L^2}$ behaves like $C^n$ for some positive constant $C>0$ that can be numerically estimated, and then one can conclude that the power series (\ref {Eq5}) converges when $|\nu | < \nu^{\star} := C^{-1}$. 

For instance, concerning the convergence of the power series for $E_N$ in the model with an infinite well potential we observe that
\be
e_n =(-1)^n 4\frac {a_n}{\pi^n 8^{n}} 
\ee
where
\be
a_n \sim a^n 
\ee
for large $n$ and where
\be
a \le 4 \, . 
\ee
Thus, we expect that the power series $E_N$ is absolutely convergent for any $\nu$ such that $|\nu |< \nu^\star$ where $\nu^\star = \frac {8\pi }{a} \ge 2 \pi$.\ Similarly, a numerical estimate of the radius of convergence for the harmonic potential case could be obtained.

\begin{figure}[h]
\begin{center}
\includegraphics[height=6cm,width=6cm]{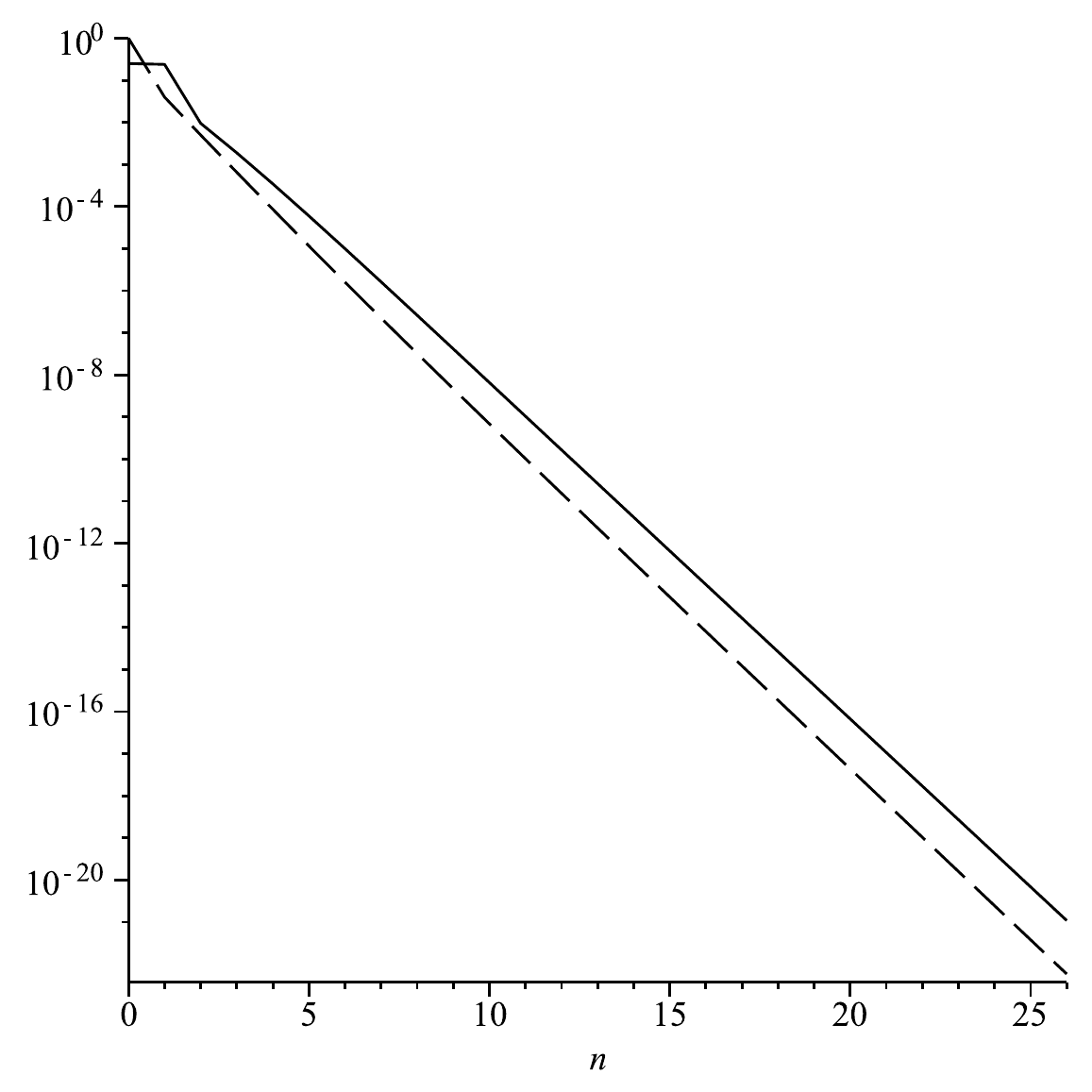}
\includegraphics[height=6cm,width=6cm]{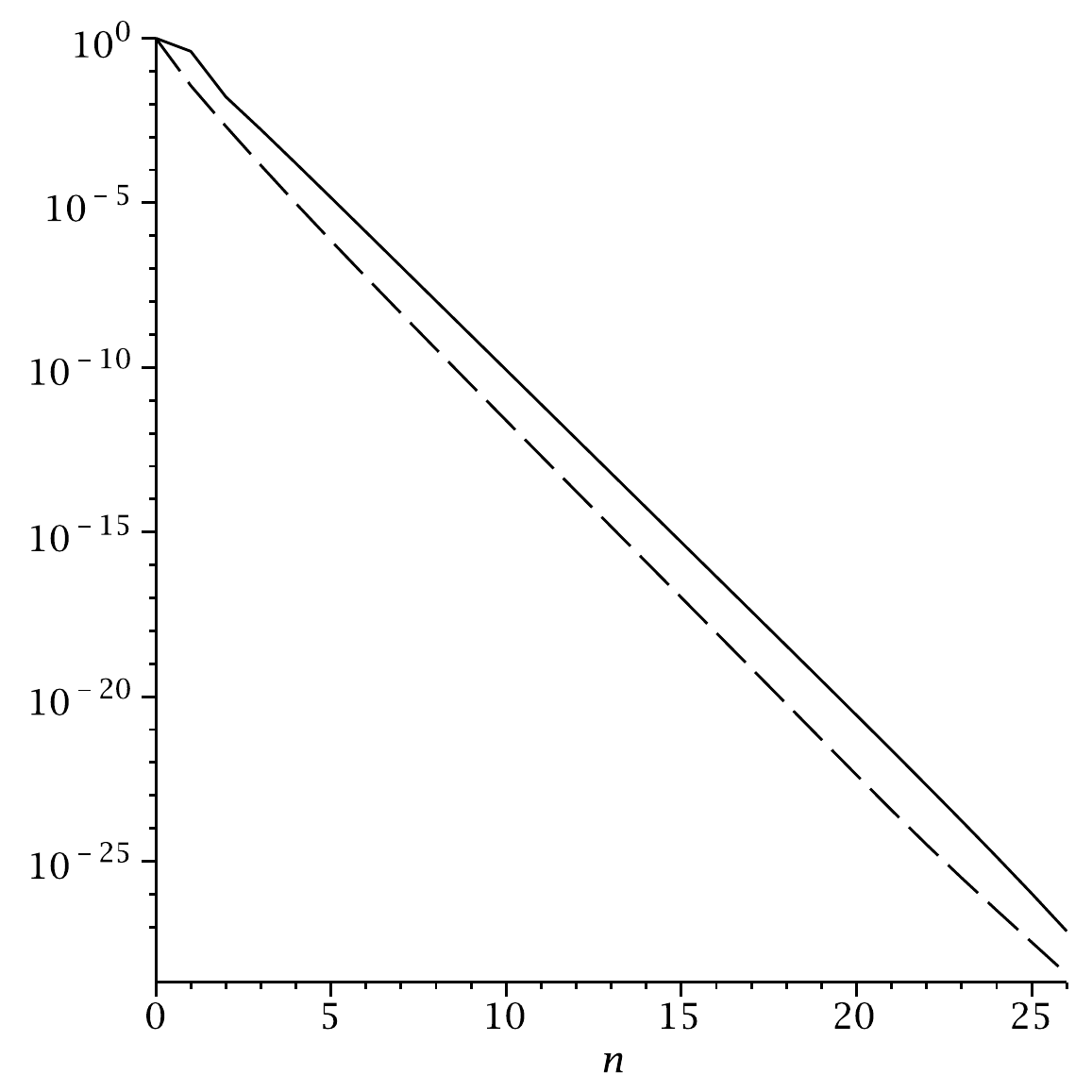}
\caption {In the two figures are respectively plotted the values of $|e_n|$ (full lines) and $\| \phi_n \|_{L^2} $ (broken lines) for the infinite well potential case (left hand side plot) and for the harmonic potential case (right hand side plot).}\label {fig1}
\end{center}
\end{figure}

\appendix

\section {A simple estimate} \label {AppA}

Let
\be
J := J(n)= \sum_{m=1}^{n-1} \frac {1}{(m+1)^2 (n-m+1)^2}\, , \ n >2 \, .
\ee
A simple inequality gives that
\be
J(n) \le 2 \int_1^{(n+1)/2} \frac {1}{x^2(n+1-x)^2} dx = 2\frac {n^2-1+2n\ln (n)}{(n+1)^3n} = \frac {f(n)}{(n+1)^2}
\ee
where
\be
f(n):= 2\frac {n^2-1+2n\ln (n)}{(n+1)n} \le 2.70 .
\ee
In fact, such an estimate is not optimal. \ A simple numerical experiment shows that 
\be
J(n) = \frac {g(n)}{(n+1)^2} \ \mbox { where }  g(n) \le g(19)= 1.517106786\, . 
\ee
Furthermore, a closed expression for $J(n)$ could be given by means of Polygamma functions; however, we don't dwell here on this detail.


\begin{thebibliography}{99}

\bibitem {A} A.Ambrosetti, M.Badiale, and  S.Cingolani, {\it Semiclassical States of Nonlinear Schr\"odinger Equations}, 
Arch. Rational Mech. Anal. {\bf 140}, 285-300 (1997).

\bibitem {An1} J.\'Angy\'an, and P.R.Surj\'an, {\it Normalization corrections to perturbation theory for the time-independent nonlinear Schr\"odinger equation}, Phys. Rev. A {\bf 44}, 2188-2191 (1991).

\bibitem {An2} J.\'Angy\'an, {\it Rayleigh-Schr\"odinger Perturbation
Theory for Nonlinear Schr\"odinger Equations with Linear Perturbation}, International Journal of Quantum Chemistry {\bf 47}, 469-483 (1993).

\bibitem {Antoine} X. Antoine, W. Baoc, and C. Besse, {\it Computational methods for the dynamics of the nonlinear
Schr\"odinger/Gross–Pitaevskii equations}, Computer Physics Communications {\bf 184}, 2621-2633 (2013).

\bibitem {Aschbacher} W.H. Aschbacher, J. Fr\"ohlich, G M. Graf, K.Schnee, and  M.Troyer, {\it Symmetry breaking regime in the nonlinear Hartree equation}, J. Math. Phys. {\bf 43} 3879-3891 (2002).

\bibitem {Bao1} W. Bao, Y. Cai, and Y. Feng,{\it Improved uniform error bounds of the time-splitting methods for the long-time (nonlinear) Schr\"odinger equation}, Math. Comp. {\bf 92}, 1109-1139 (2023).

\bibitem {Bao2} W. Bao, and C. Wang, {\it Optimal error bounds on the exponential wave integrator for the nonlinear Schr\"odinger equation with low regularity potential and nonlinearity}, arXiv:2302.09262:1-23 (2023).

\bibitem {C} E.Canc\`es, and C.Le Bris, {\it On the perturbation methods for some nonlinear quantum chemistry models}, Math. Mod. and Meth. in App. Sci. {\bf 8}, 55-94 (1998).  

\bibitem {Cordero} D. Cordero-Erausquin, B. Nazaret, and C. Villani, {\it A mass-transportation approach to sharp
Sobolev and Gagliardo-Nirenberg inequalities}, Advances in Mathematics {\bf 182} 307-332 (2004).

\bibitem {D} H.T. Davis, {\it Introduction to Nonlinear Differential and Integral Equations}, Dover Books on Mathematics (1962).

\bibitem {DellaCasa} F. Della  Casa, and A. Sacchetti, {\it Stationary states for non linear one-dimensional Schr\"odinger equations with singular potential}, Physica D: Nonlinear Phenomena {\bf  219} 60-68 (2006).

\bibitem {F} S.Fishman, Y. Krivolapov, and A.Soffer, {\it Perturbation theory for the nonlinear Schr\"odinger equation with a random potential}, Nonlinearity {\bf 22}, 2861-2887 (2009).

\bibitem {Fl} A.Floer, and A.Weinstein, {\it Nonspreading wave packets for the cubic Schr\"odinger equation with a bounded potential}, J. Funct. Anal. {\bf 69}, 397-408 (1986).

\bibitem {Gr} M. Grillakis, {\it Linearized instability for nonlinear Schr\"odinger and Klein-Gordon equations}, Commun.
 Pure Appl. Math. {\bf 41} 745-774 (1988).

\bibitem {K} J.Killingbeck, and G.Jolicard, {\it A numerical method for the nonlinear oscillator problem}, 
Chemical Physics Letters {\bf 284}, 359-362 (1998).

\bibitem {O} Y.G.Oh, {\it Existence of semiclassical bound states of nonlinear Schr\"odinger equations
with potential in the class $(V)_a$}, Comm. Partial Diff. Eqs. {\bf 13}, 1499-1519 (1988).

\bibitem {P} D.E. Pelinovsky, {\it Localization in Periodic Potentials. \ From Schr\"odinger Operators to the Gross-Pitaevskii Equation}, Cambridge University Press (2011).

\bibitem {RW} H.Rose, and M.I.Weinstein, {\it On the bound states of the nonlinear Schr\"odinger equation with a linear potential}, Phys. D {\bf 30}, 207-218 (1988).

\bibitem {Sa} A.Sacchetti, {\it Universal Critical Power for Nonlinear Schr\"odinger Equations
with a Symmetric Double Well Potential}, Phys. Rev. Lett. {\bf 103}, 194101:1-4 (2009).

\bibitem {Sacchetti} A. Sacchetti, {\it Spectral splitting method for nonlinear Schr\"odinger equations with quadratic potential},  Journal of Computational Physics {\bf 459} 111154:1-18 (2022).

\bibitem {SW1} A.Soffer, and M.I.Weinstein, {\it Multichannel nonlinear scattering for nonintegrable equations}, Commun. Math. Phys. {\bf 133}, 119-146 (1990).

\bibitem {SW2} A.Soffer, and M.I.Weinstein, {\it Selection of the ground state for nonloinear Schr\"odinger equations}, Rev. Math. Phys. {\bf 16} 977-1071 (2004).

\bibitem {Soffer} A. Soffer, and X. Zhao, {\it On multichannel solutions of nonlinear Schr\"odinger equations: algorithm, analysis and numerical explorations}, J. Phys. A: Math. Theor. {\bf 48} 135201:1-23 (2015).

\bibitem {Su} C.Sulem, and P.-L. Sulem, {\it The Nonlinear Schr\"odinger Equation: 
Self-Focusing and Wave Collapse}, Springer-Verlag (1999).

\bibitem {S} P.R.Surj\'an, and J.\'Angy\'an, {\it Perturbation theory for nonlinear time-independent Schr\"odinger equations}, Phys. Rev. A {\bf 28}, 45-48 (1983).

\bibitem {V} A.Vannucci, P.Serena, and A.Bononi, {\it The RP Method: A New Tool for the Iterative Solution
of the Nonlinear Schr\"odinger Equation}, 
Journal of Lightwave Technology {\bf 20}, 1102-1112 (2002).

\bibitem {W} X.Wang, {\it On concentration of positive bound states of nonlinear Schr\"odinger equations}, Comm. Math. Phys. {\bf 153}, 223-243 (1993).

\end{thebibliography}
\end{document}